\documentclass[prl,twocolumn,showpacs,superscriptaddress,nofootinbib]{revtex4-1}

\usepackage[table, svgnames, dvipsnames]{xcolor}
\usepackage{makecell, cellspace}
\usepackage{bbold}
\usepackage{amsmath, dsfont}
\usepackage{amsfonts}
\usepackage{bbm}
\usepackage{mathdesign}
\usepackage{mathtools}
\usepackage{amsthm}
\usepackage{bm}
\usepackage{amssymb}
\usepackage[normalem]{ulem}
\usepackage{braket}
\usepackage{hyperref}
\usepackage{tikz}
\usepackage{ifthen}
\usepackage{stmaryrd}
\usetikzlibrary{tikzmark}
\usepackage{upgreek}
\usepackage{braket}
\usepackage{physics}

\usepackage[]{lineno} 

\usepackage{graphicx}
\graphicspath{ {./images/} }

\usepackage[shortlabels]{enumitem}
 
\usepackage{tikz-cd}
\usepackage{systeme}
\usepackage{bm}
\usepackage{tensor}
\usepackage{xparse}
\usepackage{amsmath}
\usepackage{physics}
\usepackage{amsthm}
\usepackage{mathtools}
\usepackage{amssymb}
\usepackage{dsfont}
\usepackage{xcolor}
\usepackage{upgreek}
\usepackage{mathrsfs}
\usepackage{mathrsfs}

\usepackage{hyperref}
\hypersetup{
    colorlinks=true,
    linkcolor=blue,
    filecolor=red,      
    urlcolor=teal,citecolor=purple
}

\usepackage{lettrine}
\usepackage{setspace}

\usepackage{wrapfig}

\theoremstyle{plain}

\newtheorem{prop}{Proposition}
\newtheorem{lem}[prop]{Lemma}
\newtheorem{cor}[prop]{Corollary}

\newtheorem{remark}[prop]{Remark}

\theoremstyle{remark}

\renewcommand{\L}{\mathcal{L}}

\newcommand{\B}{\mathcal{B}}

\newcommand{\OO}{\mathcal{O}}

\newcommand{\K}{\mathcal{K}}

\newcommand{\A}{\mathcal{A}}

\newcommand{\M}{\text{M}}

\newcommand{\brad}[1]{( #1|}
\newcommand{\ked}[1]{|#1)}
\newcommand{\braked}[2]{( #1 | #2 )}
\newcommand{\kedbra}[2]{|#1 )( #2|}

\DeclarePairedDelimiterX\hej[1]\lbrace\rbrace{#1}

\DeclarePairedDelimiterX{\braketHS}[2]{\langle}{\rangle_{\textsc{hs}}}{#1|#2}
\DeclarePairedDelimiterX{\expect}[1]{\langle}{\rangle}{#1}

\DeclarePairedDelimiterX{\inner}[2]{\langle}{\rangle}{#1, #2}
\DeclarePairedDelimiterX{\iner}[1]{\langle}{\rangle}{#1, #1}

\begin{document}
\title{Ultimate Speed Limits to the Growth of Operator Complexity}

\author{Niklas H{\"o}rnedal}
\affiliation{Department  of  Physics  and  Materials  Science,  University  of  Luxembourg,  L-1511  Luxembourg, G. D.  Luxembourg}
\affiliation{Fysikum, Stockholms Universitet, 106 91 Stockholm, Sweden}
\email{niklas.hornedal@uni.lu}

\author{Nicoletta Carabba}
\affiliation{Department  of  Physics  and  Materials  Science,  University  of  Luxembourg,  L-1511  Luxembourg, G. D.  Luxembourg}

\author{Apollonas S. Matsoukas-Roubeas}
\affiliation{Department  of  Physics  and  Materials  Science,  University  of  Luxembourg,  L-1511  Luxembourg, G. D.  Luxembourg}

\author{Adolfo del Campo}
\affiliation{Department  of  Physics  and  Materials  Science,  University  of  Luxembourg,  L-1511  Luxembourg, G. D.  Luxembourg}
\affiliation{Donostia International Physics Center,  E-20018 San Sebasti\'an, Spain}


\begin{abstract}
\textbf{\abstractname:}
In an isolated system, the time evolution of a given observable in the Heisenberg picture can be efficiently represented in Krylov space. In this representation, an initial operator becomes increasingly complex as time goes by, a feature that can be quantified by the Krylov complexity. We introduce a fundamental and universal limit to the growth of the Krylov complexity by formulating a Robertson uncertainty relation, involving the Krylov complexity operator and the Liouvillian, as generator of time evolution. We further show the conditions for this bound to be saturated and illustrate its validity in  paradigmatic models of quantum chaos.
\end{abstract}

\maketitle

\vspace{1cm}

\noindent\textbf{\textsf{Introduction}}
\lettrine[lraise=0.1, nindent=0em, slope=-.5em]{Q}{uantum} speed limits (QSL) impose fundamental constraints on the pace at which a physical process can unfold.
Since their conception \cite{MT45,ML98}, they have been formulated as bounds on the minimal time at which a distance between quantum states can be traversed.
The freedom in the choice of the distance can be used to sharpen the discrimination between quantum states, and with it, the notion of the speed of evolution \cite{Pires16,Campaioli18}. 
Additional efforts have been devoted to exploring the role of the underlying dynamics, generalizing early results from isolated systems to open \cite{Taddei13,delcampo13,DeffnerLutz13,Campaioli2019} and classical processes \cite{Shanahan18,Okuyama18}. The resulting speed limits have become a useful tool in various branches of physics, ranging from information processing \cite{Lloyd00}  to many-body physics \cite{Bukov19}, quantum control \cite{Caneva09} and  quantum metrology \cite{Giovannetti11}. However, traditional QSL are too conservative in estimating the relevant time scales in many processes, such as thermalization \cite{Eisert15}. This has motivated the development of speed limits suited for specific measures and observables \cite{Nicholson20}, as in the pioneering work by Mandelstam and Tamm \cite{MT45}. In this sense, certain speed limits follow from generalized uncertainty relations such as those derived by Heisenberg and Robertson \cite{Braunstein96}.

In parallel with the study of QSL, quantifying the complexity of a physical process is a central task for the advancement of fundamental physics and quantum technologies. 
 Lloyd pointed out that the computational complexity of physical processes is limited by QSL \cite{Lloyd02}. Analogously, the circuit complexity of a quantum state  \cite{Susskind16}, defined as the number of elementary operations required to generate it from a reference state, can be characterized in terms of conventional QSL \cite{Brown16,Brown16prd,Chapman18,MolinaVilaplana18}. 
A complementary approach for many-body quantum systems focuses on the buildup of  complexity in the time-evolution of an initial  local observable, known as
operator growth \cite{keyserlingk2018,khemani2018,nahum2018,gopalakrishnan2018,rakovsky2018}. 
The intuition is that simple operators unitarily evolve into increasingly complex ones. Quantum information initially encoded in a few degrees of freedom is thus scrambled over the system in the course of evolution, making it impossible to recover it through local measurements and giving rise to thermalization. The unambiguous description of this scrambling process remains an open problem. One possibility is to probe it via an out-of-time-ordered correlator \cite{larkin1969,maldacena2016} that may be used to identify an analog of the Lyapunov exponent, providing a connection with classical chaos, e.g., the butterfly effect.
Such quantum Lyapunov exponent obeys a universal upper bound \cite{maldacena2016}, which helps refine the notion of maximal chaos, is saturated by black holes, and is further tied to the eigenstate thermalization hypothesis \cite{murthy2019,srednicki1994}.
A related approach, which we shall pursue in this work, is to study the dynamical evolution of operators in Krylov space, exploited in numerical techniques such as the recursion method \cite{Viswanath94}. In this context,  operator growth is quantified by the so-called Krylov complexity, a measure of the delocalization of the time-dependent operator in the Krylov basis \cite{parker2019,Barbon2019,rabinovici2021,dymarsky2021,Jian21}. 
The authors of \cite{parker2019} made a conjecture on the universal operator growth, namely, that Krylov complexity can grow at most exponentially, and it does so in generic non-integrable systems. Remarkably, its growth rate upper bounds the Lyapunov exponent, establishing a connection with the bound on out-of-time-ordered correlators \cite{maldacena2016,Dymarsky20}. Further studies have shown that exponential operator growth is possible in free and integrable systems \cite{DymarskySmolkin21},  while the role of the interaction graph in a quantum network has been explored in \cite{OlleRosa22}.  

Here, we characterize the growth of Krylov complexity by deriving a fundamental limit on its rate of change and by studying analytically the conditions under which this bound is saturated. Our results show that saturation, which is also found to correspond to a particular notion of minimum uncertainty, occurs whenever the dynamical evolution of the system has the underlying structure of a three-dimensional \textit{complexity algebra}, which was introduced by \cite{caputa2021}. In this setting, the unitary evolution of an operator can be represented as the displacement of generalized coherent states \cite{caputa2021},  which display classical-like behavior \cite{perelomov1986}.
As demonstrated in several paradigmatic examples, the saturation of the growth rate may be possible in some chaotic systems, but quantum chaos is not required for it.
\\ \\
\noindent\textbf{\textsf{Results and Discussion}}\\
\noindent\textbf{\textsf{\small Quantum dynamics in Krylov space}}\\
\noindent Consider an isolated quantum system in which the time evolution of an observable $\mathcal{O}$  is generated by a time-independent Hamiltonian $H$ according to the Heisenberg equation of motion $\partial_t \mathcal{O}(t) = i[H,\mathcal{O}(t)]$, setting $\hbar=1$. The solution to this equation with the initial condition $\mathcal{O}(0) = \mathcal{O}$ is given by $\mathcal{O}(t) = e^{itH}\mathcal{O} e^{-itH}$. In terms of the Liouvillian superoperator given by $\L = [H,\cdot]$, the Taylor expansion of the time-evolving observable $\mathcal{O}(t) = \sum_{n=0}^\infty \frac{(it)^n}{n!}\L^n\mathcal{O}$ shows that its dynamics is contained in the complex linear span of the operators $\{\L^n\mathcal{O}\}_{n=0}^\infty$. This span is   completely determined by the Hamiltonian and the initial observable and  is known as the Krylov space. 

From now on, we consider the restriction of each operator and superoperator to the Krylov space. To highlight the vector space structure, we  make use of the braket notation $\ked{A}$ when expressing an operator $A$ in an equation. We choose to equip the Krylov space with an inner product satisfying the properties 
\begin{enumerate}
\label{inner list}
\item $\braked{A}{\L B} = \braked{\L A}{B}$, $\forall A,B$.
\item $\braked{A}{\L A} = 0$, when $A$ is Hermitian.
\end{enumerate}
An example of a family of inner products satisfying these two properties is given by $\braked{A}{B} = \expect{e^{ \beta H/2}A^\dagger e^{- \beta H/2 }B}_{\beta}$. The bracket $\expect{\cdot}_\beta$  denotes the thermal expectation value with respect to the equilibrium Gibbs state $e^{-\beta H}/Z$ and thus $\braked{A}{B} $ reduces to  the Hilbert–Schmidt inner product when $\beta=0$, up to a normalization factor. It follows from the second property of the inner product that the operators $\mathcal{O}$ and $\L\mathcal{O}$ are orthogonal. Let $b_0 = \norm{\mathcal{O}}$ and $b_1 = \norm{\L\mathcal{O}}$, where $\norm{\cdot}$ is the norm induced by the inner product. By starting from the normalized vectors $\mathcal{O}_0 = \mathcal{O}/b_0$ and $\mathcal{O}_1 = \L\mathcal{O}/b_1$, we can construct an orthonormal basis $\{\mathcal{O}_n\}_{n=0}^{D-1}$ for the Krylov space by applying the Lanczos algorithm. This algorithm works as follows: given the first $n+1$ basis vectors, one constructs the orthogonal vector $\ked{A_{n+1}} = \L\ked{\mathcal{O}_n}-b_{n}\ked{\mathcal{O}_{n-1}}$, where $b_n = \norm{A_n}$ and then normalize it to obtain $\ked{\mathcal{O}_{n+1}}$. We call the constructed basis \textit{the Krylov basis}. It is possible that the Krylov dimension $D$ is infinite, in which case the Lanczos algorithm never halts. We remark that the Lanczos algorithm is only guaranteed to construct an orthonormal basis if the Liouvillian is self-adjoint, i.e., the first property of the inner product is satisfied. Generally, the Lanczos algorithm  involves a third term on the right-hand side of the equation for $\ked{A_{n+1}}$. This term is however always zero whenever the second property of the inner-product is satisfied. Thus, with our chosen inner-product, the action of the Liouvillian on the Krylov basis takes the specific form $\L\ked{\mathcal{O}_n} = b_{n+1}\ked{\mathcal{O}_{n+1}} + b_n\ked{\mathcal{O}_{n-1}}$. As  pointed out in \cite{caputa2021}, this motivates one to consider abstract raising and lowering operators that we denote by $\L_+$ and $\L_-$, respectively. Their action on the Krylov basis is given by $\L_+\ked{\mathcal{O}_n} = b_{n+1}\ked{\mathcal{O}_{n+1}}$ and $\L_-\ked{\mathcal{O}_n} = b_{n}\ked{\mathcal{O}_{n-1}}$. The Liouvillian can then  be expressed as their sum.

It is further convenient to introduce the real valued functions $\varphi_n(t)$, which appear in the expansion of $\mathcal{O}(t)$ as $\ked{\mathcal{O}(t)} = \frac{1}{\norm{\mathcal{O}}}\sum_{n=0}^{D-1}i^n\varphi_n(t)\ked{\mathcal{O}_n}$.  We will refer to these functions as \textit{the amplitudes} of the observable. These amplitudes evolve according to the recursion relation $\partial_t\varphi_n(t) = b_{n}\varphi_{n-1}(t) - b_{n+1}\varphi_{n+1}(t)$ with the initial conditions $\varphi_0(0) = 1$ and $\varphi_n(0) = 0$ for $n>0$. Thinking of the Krylov basis vectors as forming the sites of a one dimensional lattice,   $b_n$ can be interpreted as a hopping amplitude, see, e.g.,  \cite{parker2019,Barbon2019}. In this sense, one can think of $\mathcal{O}$ as a one dimensional discrete wave function that is initially localized and then spreads out over the lattice as time evolves. An increase in the population of the sites further away from the origin reflects a greater increase of complexity of the observable. In order to quantify this, it is natural to consider the Krylov complexity of $\mathcal{O}(t)$, defined to be
\begin{equation}
    K(t) = \sum_{n=0}^{D-1} n\abs{\varphi_n(t)}^2.
\end{equation}
The main task of our work is to bound the growth of Krylov complexity. Due to unitary dynamics, the norm of the evolution is preserved and the Krylov complexity is unchanged if one normalizes the operators studied. We will, therefore, without loss of generality, consider $\mathcal{O}$ to be normalized.
By introducing the complexity operator $\K = \sum_{n=0}^{D-1} n\kedbra{\mathcal{O}_n}{\mathcal{O}_n}$, which plays the role of the position operator in the Krylov lattice,  it is possible to express Krylov complexity as the ``expectation value'' of $\K$ with respect to $\mathcal{O}(t)$. More precisely, if $\expect{\K}_t \equiv \braked{\mathcal{O}(t)}{\K\mathcal{O}(t)}$ then $K(t) = \expect{\K}_t$.

\smallskip

\noindent{\textbf{\textsf{\small Dispersion bound on Krylov complexity}}}\\
\noindent If the Krylov space forms an inner product space in which $\A$ and $\B$ are self-adjoint superoperators, then there ought to exist a Robertson uncertainty relation given by $\Delta\A\Delta \B \geq \frac{1}{2}\abs{\expect{[\A,\B]}}$, where $\Delta \A = \sqrt{\expect{\A^2}-\expect{\A}^2}$ is the dispersion of $\A$ with respect to some state $\ked{A}$. When the Krylov dimension is infinite it is necessary that $\ked{A}$ is contained in the intersection between the domains of $\A\B$ and $\B\A$, otherwise the inequality might not hold \cite{Davidson1965}. Letting $A = \mathcal{O}(t)$, $\A = \L$, $\B = \K$ and noting that $\Delta\L = b_1$, we can rewrite the uncertainty relation as 
\begin{equation}
\label{dispersion bound}
    \abs{\partial_tK(t)} \leq 2b_1\Delta\K.
\end{equation}
In other words, the growth of Krylov complexity is upper bounded by a constant times the dispersion of the complexity operator. 
By defining a characteristic time-scale $\tau_K=\Delta\K/ \abs{\partial_tK(t)}$, one obtains  $\tau_Kb_1\geq1/2$ which takes the form of a Mandelstam-Tam bound, and emphasizes the role of $b_1= \norm{\L\mathcal{O}}$ as a norm of the generator of evolution in Krylov space.
To avoid confusion with the uncertainty relation for observables, we will refer to this bound as \textit{the dispersion bound}. We note that no bound tighter than \eqref{dispersion bound} can be found by considering the more general Schr{\"o}dinger uncertainty relation, as  the extra term given by the anti-commutator  identically vanishes, as shown in Methods.

It is not self-evident that saturation of the dispersion bound can be achieved under unitary dynamics of the observable. 
There are very specific relations between $\L$, $\mathcal{O}$ and $\K$ that need to hold: the Liouvillian is required to be tridiagonal in the eigenbasis of the complexity operator and the initial state of the observable is required to be parallel to the eigenvector with the lowest eigenvalue. The conditions for the saturation of the dispersion bound are thus highly constrained and differ from those known for saturation of a Robertson uncertainty relation in general. 
The required conditions admit a geometrical interpretation, elaborated in Methods. The bound is saturated if and only if the evolution curve moves along the gradient of the Krylov complexity. This requires that the dynamics is directed along the direction that maximizes the local growth of complexity; see Methods. The only exception involves extremal points in which any direction away from the extremal point leads to saturation. This is indeed the case for $t=0$. 
Indeed,  there exists Liouvillians of the form $\L = \L_+ + \L_-$ for which the tangent of the generated path will be parallel with the gradient for all times. 

\smallskip

\noindent{\textbf{\textsf{\small Saturation of the dispersion bound}}}\\
\noindent  Time evolutions saturating the dispersion bound are characterized by  a unique algebraic structure.  Define the superoperator $\B = \L_+ -\L_-$. Following \cite{caputa2021}, we consider their {\it simplicity hypothesis}: namely, the assumption that $\L$, $\B$ and the commutator $\Tilde{\K} = [\L, \B]$ close an algebra with respect to the Lie bracket. It was shown in \cite{caputa2021} that this forces $\Tilde{\K}$ to be related to the complexity operator via $\Tilde{\K} = \alpha\K+\gamma$, where $\alpha,\gamma\in\mathbb{R}$. We show in Supplementary Note 2 that $\gamma$ is a positive number and $\alpha$ is a real number satisfying the condition $\alpha\geq 0$ for infinite Krylov dimension and $\alpha = -\frac{2\gamma}{D-1}$ for finite Krylov dimension. Moreover, the only possible closure of the algebra is given by the commutation relations
\begin{equation} \label{algebra closure}
        [\L, \B] = \Tilde{\K},\quad [\Tilde{\K}, \L] = \alpha\B,\quad[\Tilde{\K}, \B] = \alpha\L.
\end{equation}
%
  \begin{figure}[t]
  \centering
\includegraphics[width=1\linewidth]{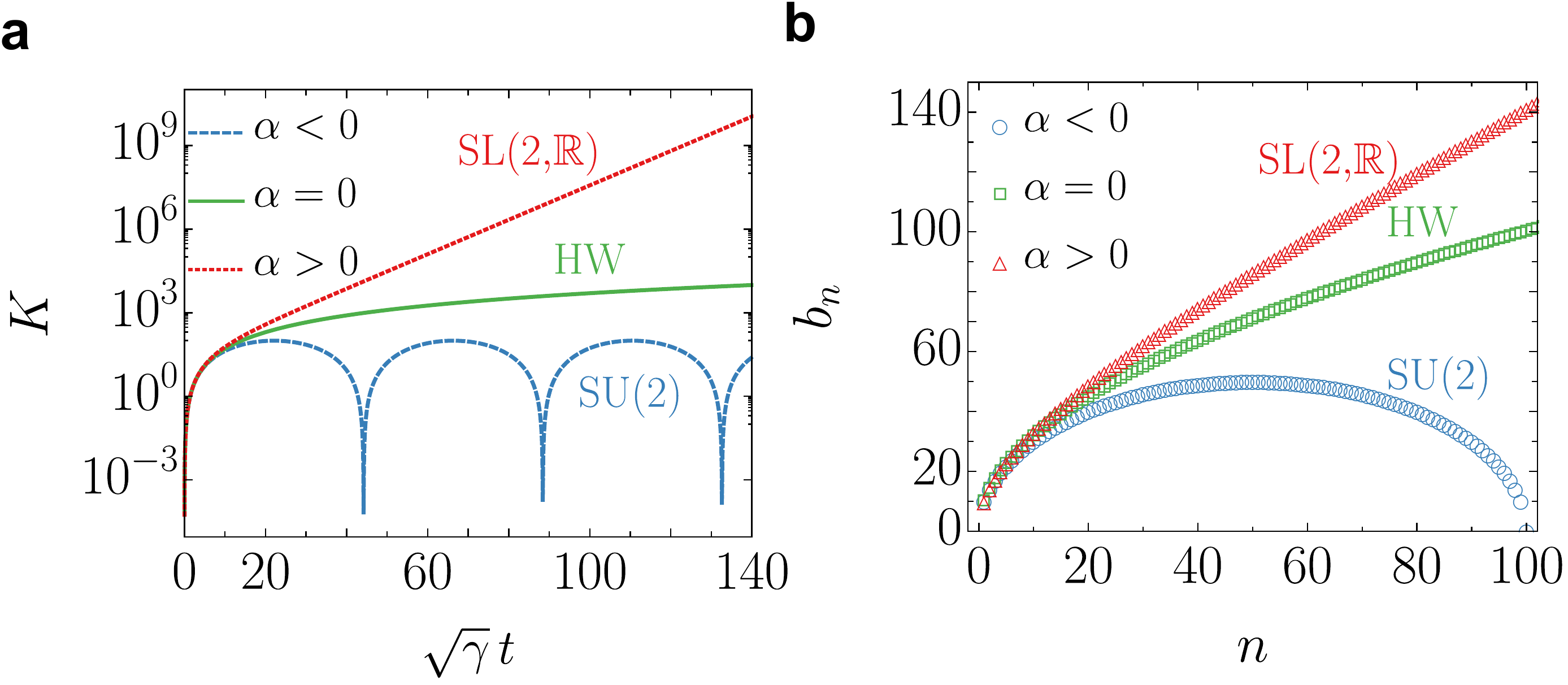}
\caption{{\bf Growth of  Krylov complexity at the speed limit.}
Saturation of the dispersion bound occurs in three different scenarios, each of which is associated with a different complexity algebra, that is specified by the sign of $\alpha$.   
{\sf{\textbf{a}}} Time-dependence of the Krylov complexity. 
{\sf{\textbf{b}}} The  corresponding growth of the Lanczos coefficients  in the Krylov lattice.   
The plots are representative of  the three different scenarios. The Krylov dimension in the $\textrm{SU}(2)$ case is $D=100$, and  infinite in all other cases.  In {\sf{\textbf{b}}} we choose $\alpha=4$ and $-4$ for the $\textrm{SL}(2,\mathbb{R})$ and $\textrm{SU}(2)$ algebras, respectively, while $\alpha$ is always zero in the $\textrm{HW}$ case. Finally, the parameter $\gamma$ in {\sf{\textbf{b}}} is chosen in each case such that the corresponding Lanczos coefficients share the same behavior near the origin of the Krylov lattice. Specifically, $\gamma= 202,\, 200,\, 198$ for the $\textrm{SL}(2,\mathbb{R})$, $\textrm{HW}$ and $\textrm{SU}(2)$ algebras, respectively. 
   }
  \label{SatDBound}
    \end{figure} 
Given this algebra, the evolving observable can be interpreted as a curve of generalized coherent states evolving according to the displacement operator $D(\xi) = e^{\xi\L_+ -\overline{\xi}\L_-}$, where $\xi = it$. Moreover, the initial state is the highest weight state of the representation, which is annihilated by $\mathcal{L_-}$ by construction. Coherent states can be viewed as the states closest to the classical ones in the sense that they typically minimize an uncertainty relation. It is for example known that coherent states of the Harmonic oscillator saturate the Robertson uncertainty relation for the pair of observables of position and momentum. Building on this intuition, we could expect that the dispersion bound is saturated for the simplicity hypothesis. It turns out that this intuition is indeed correct. In fact, as we show in  Supplementary Note 2, the dispersion bound is saturated if and only if the simplicity hypothesis holds. The saturation of the dispersion bound dictates the evolution of the Krylov complexity, where three different scenarios are possible, as shown in Figure \ref{SatDBound}{\sf{\textbf{a}}}.  The growth of complexity at the speed limit is described by the differential equation 
\begin{equation}
\partial_t^2K(t) = \alpha K(t) + \gamma,
\end{equation}
 with the conditions that $K(0)=0$ and $K(-t) = K(t)$.
 For finite Krylov dimension, saturation of the dispersion bound sets  the complexity growing according to $K(t) = (D-1)\sin^2{\omega t}$, where $\omega = \sqrt{\frac{\gamma}{2(D-1)}}$. In this, case, the corresponding complexity algebra \eqref{algebra closure} reduces to the $\textrm{SU}(2)$ algebra. By contrast, for infinite Krylov dimension there are two distinct scenarios for the complexity growth: for $\alpha>0$ one finds $K(t) = \frac{2\gamma}{\alpha}\sinh^2{\frac{\sqrt{\alpha} t}{2}}$, while for $\alpha=0$ the solution reads $K(t) = \frac{\gamma}{2}t^2$. The complexity algebra in these two cases reduces to $\textrm{SL}(2,\mathbb{R})$ and the Heisenberg-Weyl algebra (\textrm{HW}), respectively. Reference examples maximizing the Krylov-complexity growth rate at all times are discussed in Supplementary Note 1. One such example with $\alpha>1$ is the Sachdev-Ye-Kitaev (SYK) model \cite{Sachdev93}, a paradigm of quantum chaos. However, the saturation of the bound does not require quantum chaos and can indeed be achieved by a single qubit, with $\alpha=0$ (Supplementary Note 1). 
Together with the time-dependence of $K(t)$ and the complexity algebra, the value of $\alpha$ also determines the growth of the Lanczos coefficients in the Krylov lattice.
As proven in Supplementary Note 2,  the dispersion bound is saturated if and only if the Lanczos coefficients grow according to
\begin{equation}
\label{Lanczos growth}
b_n = \sqrt{\frac{1}{4}\alpha n(n-1) + \frac{1}{2}\gamma n},
\end{equation}
exhibiting  three different scalings as function of $\alpha$, see Figure \ref{SatDBound}{\sf{\textbf{b}}}. That the simplicity hypothesis implies \eqref{Lanczos growth} has already been pointed out in \cite{caputa2021}.
 For $\alpha>1$ and large $n$, this dependence captures the linear growth $b_n=\sqrt{\alpha} n$ conjectured by Parker {\it et al.} to hold in generic non-integrable systems, maximizing the Krylov complexity growth \cite{parker2019}.

\smallskip

\noindent{\textbf{\textsf{\small Krylov complexity in generic systems}}}\\
\noindent 
We next discuss the Krylov complexity growth in generic systems not fulfilling the simplicity hypothesis.
We can use Eq. (\ref{Lanczos growth}) to estimate when and at what time scale a generic  system deviates  from the bound. By expanding Krylov complexity up to fourth order we find that $K(t) = b_1^2t^2 +  \frac{1}{6}b_1^2(2b_2^2-b_1^2)t^4 + O(t^6)$. Since we can always find a value on $\alpha$ and $\gamma$ such that $b_1$ and $b_2$ satisfy \eqref{Lanczos growth}, we conclude that the bound \eqref{dispersion bound} is saturated up to third order in time. By expanding the Krylov complexity up to sixth order, we find that the Lanczos coefficient $b_3$ will appear in the last term and since we are not guaranteed to be able to find a value on $\alpha$ and $\gamma$ such that $b_1$, $b_2$ and $b_3$ satisfy\eqref{Lanczos growth}, we conclude that the system can only start deviating from the bound \eqref{dispersion bound} as a result from fifth order terms in the expansion. We can estimate this time scale by finding the value of $t$ for which the third order coefficient of $\partial_t K(t)$ is equal to its fifth order coefficient. We will call this time the \textit{deviation time}, denoted by $\tau_\text{d}$, and it is explicitly given by
\begin{equation}
\label{tbreak}
\tau_\text{d} = \sqrt{\frac{\frac{2}{3}b_1^2(2b_2^2-b_1^2)}{\frac{1}{20}b_1^2(b_1^2+b_2^2)-\frac{1}{5}b^2_2(b_1^2+b_2^2+b_3^2) +\frac{1}{2}b_2^2b_3^2}}.
\end{equation}

%
  \begin{figure}[t]
  \centering
\includegraphics[width=1\linewidth]{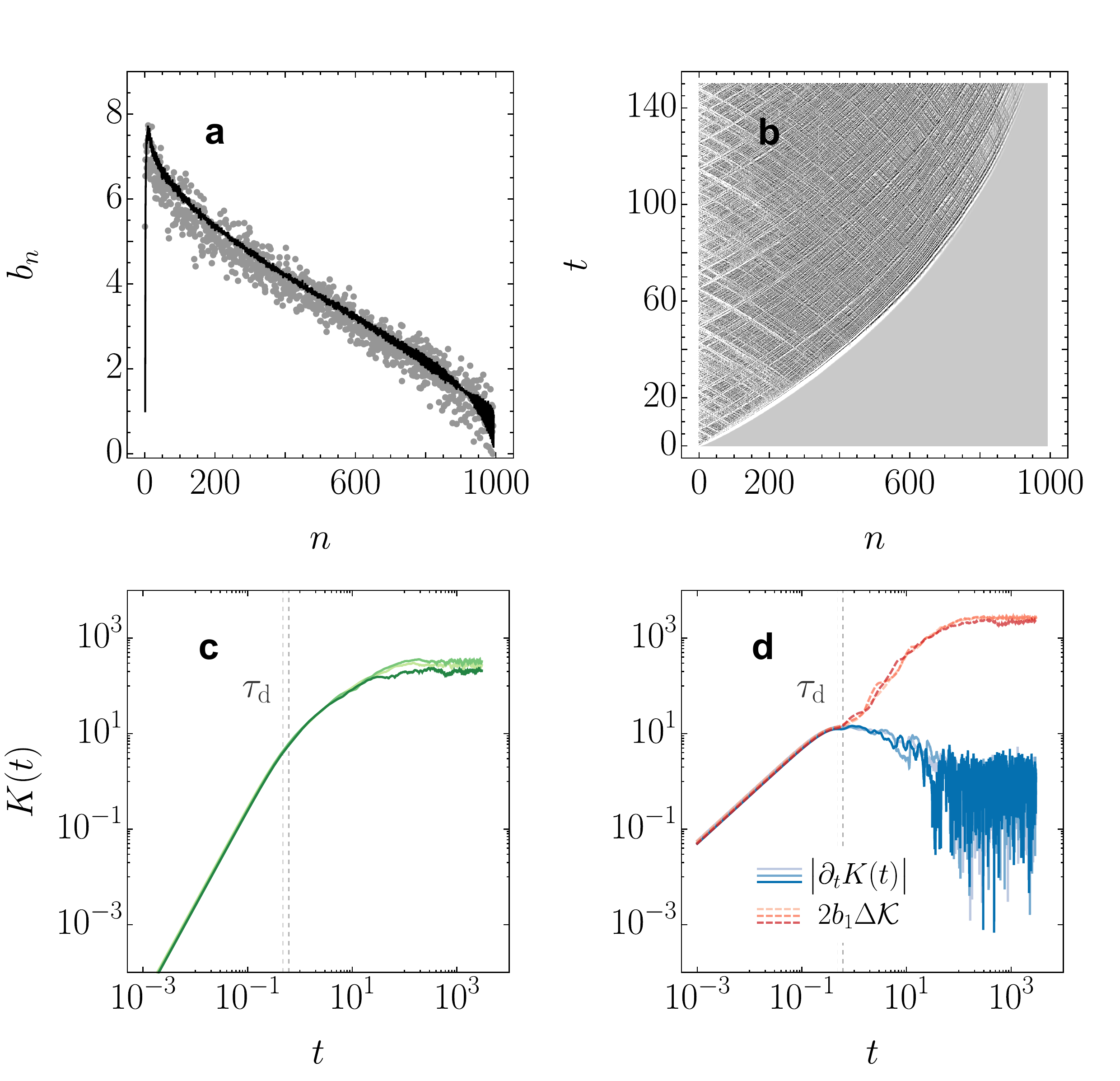}
\caption{{\bf Growth of  Krylov complexity in a generic system.}
{\sf{\textbf{a}}} Squares of the Lanczos coefficients for a single realization (gray points) and an average over $100$ random Hamiltonian matrices (black line).    
{\sf{\textbf{b}}} Operator growth in the Krylov lattice as  displayed by the dynamics of the amplitudes $|\varphi_n (t)|^2$ for a single random matrix realization.   
{\sf{\textbf{c}}} Krylov complexity (green solid lines) together with the deviation time (gray dashed line) for three independent random matrix realizations. 
{\sf{\textbf{d}}} The corresponding absolute value of  the growth rate of the Krylov complexity (blue solid lines),  together with the dispersion bound (red dashed lines),  Eq. \eqref{dispersion bound}.   
  In all figures the random Hamiltonian matrices are sampled from $\mathrm{GOE}(d)$ with standard deviation $\sigma=1$, maximal Krylov dimension $D = 993$ and a uniform initial observable operator $\mathcal{O}$. 
   }
  \label{KBcombplot}
    \end{figure} 
    

To get an understanding of the complexity growth in a generic setting, we next illustrate the Krylov dynamics of a system described by a random matrix Hamiltonian.  Specifically, we consider the Krylov complexity of an ensemble $\mathcal{E}(H)$  of random matrix Hamiltonians, a paradigm of quantum chaos \cite{Haake}.  We sample the Hamiltonian matrices $H$ from the Gaussian Orthogonal Ensemble  $\mathrm{GOE}(d)$, where $d$ is the dimension of the Hilbert space.  We then  calculate the Lanczos coefficients $\{b_n\}$ with partial re-orthogonalization \cite{simon1984,rabinovici2021}.  
Specifically,  we consider samples of real matrices $H=(X+X^\intercal)/2$,  where all elements $x \in \mathbb{R}$ of $X$ are pseudo-randomly generated with probability measure given by the normal distribution,  $\exp( -x^2/ \big(2 \sigma^2\big)) / ( \sigma \sqrt{2 \pi} )$.  
In order to study the general behaviour of Lanczos coefficients, we choose an initial observable which is represented as the normalized vector $| \mathcal{O} ) = (1/d,1/d, \dots, 1/d)^T$, expressed in a fixed eigenbasis of the Liouvillian. However, the following results do not depend strongly on the choice of $\mathcal{O}$, provided it is dense in the eigenbasis of the Hamiltonian. Figure \ref{KBcombplot}{\sf{\textbf{a}}} shows the squares of the Lanczos coefficients for a single realization  and the average $ \langle \{b_n\} \rangle_{\mathcal{E}(H)}$ over $100$ different Hamiltonians of dimension $d=32$, sampled from 
$\mathrm{GOE}(d)$ with standard deviation $\sigma=1$.
Operator growth is displayed by the time-dependent amplitudes, which are found by solving the recursion relation and exhibit diffusion-like dynamics on the Krylov basis, shown for a single realization in Figure \ref{KBcombplot}{\sf{\textbf{b}}}.
The corresponding time evolution of Krylov complexity and its growth rate are shown in panels {\sf{\textbf{c}}} and {\sf{\textbf{d}}}, respectively.
Hamiltonians sampled from GOE($d$) behave as a generic system, given that the Lanczos coefficients do not, in general, grow according to \eqref{Lanczos growth} as shown in Figure \ref{KBcombplot}{\sf{\textbf{a}}}. As a result, the growth rate starts deviating from the dispersion bound around the time scale $\tau_\text{d}$ in Eq. \eqref{tbreak}, indicated by the vertical line in  Figure \ref{KBcombplot}, panels  {\sf{\textbf{c}}} and {\sf{\textbf{d}}}. In short, while GOE Hamiltonians provide a useful paradigm in the description of quantum chaotic systems,  the dynamics generated by them does not maximize the growth of Krylov complexity for $t>\tau_\text{d}$. 

Our results establish the ultimate speed limit to operator growth in isolated quantum systems. Specifically, 
the dispersion bound  governs the growth rate of Krylov complexity, playing the role of a Mandelstam-Tamm uncertainty relation in operator space.  This bound is saturated by quantum systems in which the Liouvillian governing the time evolution fulfills a simplicity algebra. The latter arises naturally in certain quantum chaotic systems, such as the SYK model. However, other paradigmatic instances of quantum chaos, such as random-matrix Hamiltonians, do not maximize the growth of Krylov complexity. Indeed, saturation of the bound does not require quantum chaos and can be achieved, e.g.,  by a single qubit.


\noindent{\textbf{\textsf{Methods}}}\\
\noindent\textbf{\textsf{\small Vanishing of the anticommutator contribution in the Robertson uncertainty relation for $\mathcal{K}$ and $\mathcal{L}$.}}
 We establish a universal feature of Krylov complexity, valid for any physical system: namely, that its anticommutator with the Liouvillian $\mathcal{L}$ has vanishing expectation value over the evolved operator $\ked{\mathcal{O}(t)}$. The relevance of this result relies on the fact that this quantity enters the Schr\"odinger uncertainty principle for the two operators $\mathcal{K}$ and $\mathcal{L}$
\begin{align} \label{Unc}
4(\Delta \mathcal{K}\Delta \mathcal{L})^2 &\geq  
|\brad{\mathcal{O}(t) } [\mathcal{K},\mathcal{L}] \ked{ \mathcal{O}(t)}|^2 \\
&+ |\brad{\mathcal{O}(t) } \{\mathcal{K},\mathcal{L}\} \ked{ \mathcal{O}(t)}|^2,  \nonumber
\end{align}
from which one can bound the complexity rate $\partial_t K$. We have that 
\begin{equation} \label{commKL}
\brad{\mathcal{O}(t) } [\mathcal{K},\mathcal{L}] \ked{ \mathcal{O}(t)}=2i\operatorname{Im}\brad{\mathcal{O}(t) } \mathcal{K}\mathcal{L}\ked{ \mathcal{O}(t)}
\end{equation}
and
\begin{equation} \label{anticommKL}
\brad{\mathcal{O}(t) } \{\mathcal{K},\mathcal{L}\} \ked{ \mathcal{O}(t)}=2\operatorname{Re}\brad{\mathcal{O}(t) }\mathcal{K}\mathcal{L}\ked{ \mathcal{O}(t)},
\end{equation}
where
\begin{equation} \label{DiracKL}
\mathcal{K}\mathcal{L} = \sum_{n=0}^{D-1} b_{n+1} [ \, n\ked{\mathcal{O}_n}\brad{\mathcal{O}_{n+1}} + (n+1)\ked{\mathcal{O}_{n+1}}\brad{\mathcal{O}_{n}} \, ].
\end{equation}
Let us now demonstrate that the anticomutator term in Eq.~\eqref{anticommKL} is identically zero. By expanding $\ked{ \mathcal{O}(t)}$ over the Krylov basis we obtain
\begin{equation}
\begin{split}
\brad{\mathcal{O}(t) }\mathcal{K}\mathcal{L}\ked{ \mathcal{O}(t)}&= \sum_{m,n,k=0}^{D-1} (-i)^m i^k\varphi_m\varphi_k b_{n+1}(n\delta_{mn}\delta_{n+1,k}\\
&+ (n+1)\delta_{m,n+1}\delta_{nk}),
\end{split}
\end{equation}
which, by performing the sums over $k$ and $n$, yields
\begin{equation} \label{resKL}
\begin{split}
\brad{\mathcal{O}(t) }\mathcal{K}\mathcal{L}\ked{ \mathcal{O}(t)}&=i\sum_{m=0}^{D-1} m\varphi_m(\varphi_{m+1}b_{m+1}-\varphi_{m-1}b_{m})\\
&=-i\sum_{m=0}^{D-1} m\varphi_m \partial_t\varphi_m.
\end{split}
\end{equation}
Since the amplitudes $\varphi_n$ and the coefficients $b_n$ are real quantities, comparing Eqs.~\eqref{anticommKL} and \eqref{resKL} we immediately conclude that
\begin{equation} \label{vanishing anticomm}
\brad{\mathcal{O}(t) } \{\mathcal{K},\mathcal{L}\} \ked{ \mathcal{O}(t)}=0 \quad \forall\,t.
\end{equation}
Let us note that the key condition to obtain this result is the fact that the Liouvillian connects only states that are nearest-neighbors on the Krylov lattice, so that we are left with a purely imaginary phase $(-i)^m(i)^{m\pm1}=\pm i$. It is this peculiar property that allows the Liouvillian to be interpreted as a sum of generalized ladder operators $\mathcal{L}_\pm$ \cite{caputa2021}. However, let us point that here we are not making any assumption regarding the commutation rules between these operators: we are considering the structure of Krylov space in full generality. 

Moreover, from Eq.~\eqref{commKL} we immediately obtain the relation between the anticommutator $[\mathcal{K},\mathcal{L}]$ and the complexity rate $\partial_t K$: 
\begin{equation} \label{rate}
\brad{\mathcal{O}(t) } [\mathcal{K},\mathcal{L}] \ked{ \mathcal{O}(t)}=-2i\sum_{m=0}^{D-1} m\varphi_m \partial_t\varphi_m=-i\partial_t K.
\end{equation}
Therefore, the Schr\"odinger uncertainty relation ~\eqref{Unc} can be recast as the dispersion bound \eqref{dispersion bound} on the growth of Krylov complexity:
\begin{equation} \label{bound}
|\partial_t K| \leq 2 b_1 \Delta \mathcal{K}.
\end{equation}

\noindent\textbf{\textsf{\small Geometrical interpretation of the saturation of the bound.} }
For the geometrical interpretation of the saturation of the bound, we  assume the Krylov space to be of finite dimension. 
However, the results could potentially be extended to infinite-dimensional Krylov spaces as well.

The Krylov space is isomorphic to a $2D$-dimensional real vector space and we can therefore consider the Euclidean metric $g$, given by the real part of the inner product. The evolution curve of $\mathcal{O}$ will then be restricted to the unit sphere of the Krylov space. This unit sphere forms a Riemannian manifold and we can consider the Krylov complexity as a function on this manifold defined by $K(\A) = \braked{\A}{\K\A},$ for any element $\ked{\A}$ in the Krylov space with unit norm. In this sense, when we write $K(t)$ we simply mean $K(\mathcal{O}(t))$ which is consistent with how we defined complexity for the evolution. The differential of Krylov complexity will be denoted by $dK$ and its action on any tangent vector $\dot{\A}$ at $\A$ is given by $dK(\dot{\A}) = \braked{\dot{\A}}{\A} + \braked{\A}{\dot{\A}}$. This differential together with the metric can be used to define the gradient of Krylov complexity. It follows from the theory of differential geometry that the gradient of Krylov complexity at $\A$, denoted by $\nabla K(\A)$, is the unique vector satisfying the expression $g(\nabla K(\A), \dot{\A}) = dK(\dot{\A})$ for all tangent vectors $\dot{\A}$ at $\A$ \cite{Lee2018}. It can be checked that the gradient must then be given by $\nabla K(\A) = 2(\K-\expect{\K})\A$, which indeed is tangent to the unit sphere at $\A$. The change of Krylov complexity along the curve $\mathcal{O}(t)$, generated by the Liouvillian, is given by $\partial_tK(t) = g(\nabla K(t), \partial_t \mathcal{O}(t))$, where $\nabla K(t)$ is the gradient at $\mathcal{O}(t)$. Applying the Cauchy-Schwarz inequality on the right-hand side gives us the inequality
\begin{equation}
    \abs{\partial_tK(t)} \leq \norm{\nabla K(t)}\norm{\partial_t \mathcal{O}(t)}.
\end{equation}
The right-hand side of this inequality is exactly $2b_1\Delta\K$ and we note that it is saturated if and only if the tangent vector of $\mathcal{O}(t)$ is parallel to the gradient of Krylov complexity. We also note that the gradient is the zero vector at time zero and so the dispersion bound is always initially saturated.

The unitary orbit of $\mathcal{O}$ is the set of all points $U^\dagger\mathcal{O} U$, where $U$ is a unitary operator. We emphasize that this is a proper subset of the unit sphere in Krylov space which, in contrast, is the set of all points $\mathcal{U}\mathcal{O}$, where $\mathcal{U}$ is a unitary superoperator. The gradient we have considered is with respect to the unit sphere and it is therefore not obvious that this gradient will ever be tangential to the unitary orbit of $\mathcal{O}$. 
However, the gradient is indeed tangential to the unitary orbit at time zero and at all times provided the simplicity algebra is fulfilled. 

\noindent\textbf{\textsf{\small On the closure of the complexity algebra.}} 
Here we show the proof that the only possible closure of the complexity algebra introduced by \cite{caputa2021} is given by Eq.~\eqref{algebra closure}. The (anti-Hermitian) operator $\mathcal{B}=\mathcal{L}_+-\mathcal{L}_-$ ``conjugated'' to the Liouvillian can be expanded in Krylov space as
\begin{equation} \label{B}
\mathcal{B}=\sum_{n=0}^{D-1} b_{n+1} [ \, \ked{\mathcal{O}_{n+1}}\brad{\mathcal{O}_n} - \ked{\mathcal{O}_{n}}\brad{\mathcal{O}_{n+1}} \, ].
\end{equation}
We note that one can establish a formal analogy with the harmonic oscillator: $\mathcal{L}$ plays the role of the position of the harmonic oscillator, while $i\mathcal{B}$ corresponds to its momentum. However, in general the commutator between $\mathcal{L}$ and $\mathcal{B}$ is not proportional to the identity, indeed:
\begin{equation} \label{Ktilde}
\tilde{\mathcal{K}}=2[\mathcal{L}_+,\mathcal{L}_-]=2\sum_{n=0}^{D-1} (b^2_{n+1}-b^2_n) \ked{\mathcal{O}_n}\brad{\mathcal{O}_n},
\end{equation}
where it is understood that $b_0$ has to be replaced with $0$. Let us now investigate the conditions under which $\mathcal{L}$, $\mathcal{B}$ and $\tilde{\mathcal{K}}$ form a closed algebra with respect to the operation $[,]$: the so-called \textit{complexity algebra} \cite{caputa2021}. This happens if and only if the commutators $[\mathcal{L},\tilde{\mathcal{K}}]$ and $[\mathcal{B},\tilde{\mathcal{K}}]$ can be written as linear combinations of the operators $\mathcal{L}$, $\mathcal{B}$ and $\tilde{\mathcal{K}}$ themselves.  These commutators can be expanded over the Krylov basis as follows:
\begin{equation} \label{commLK}
[\mathcal{L},\tilde{\mathcal{K}}]= 2\sum_{n=0}^{D-1} f(n) b_{n+1} [ \, \ked{\mathcal{O}_{n+1}}\brad{\mathcal{O}_n} - \ked{\mathcal{O}_{n}}\brad{\mathcal{O}_{n+1}} \,],
\end{equation}
\begin{equation} \label{commBK}
[\mathcal{B},\tilde{\mathcal{K}}]= 2\sum_{n=0}^{D-1} f(n) b_{n+1} [ \,\ked{\mathcal{O}_{n+1}}\brad{\mathcal{O}_n} + \ked{\mathcal{O}_{n}}\brad{\mathcal{O}_{n+1}} \,],
\end{equation}
where we have defined
\begin{equation} \label{f(n)}
f(n)=b^2_{n+1}-b^2_n- (b^2_{n+2}-b^2_{n+1})= \frac{\tilde{\mathcal{K}}_{nn}-\tilde{\mathcal{K}}_{n+1,n+1}}{2}.
\end{equation}
Now, it is clear that the commutator \eqref{commLK} between $\mathcal{L}$ and $\tilde{\mathcal{K}}$ cannot contain any element of the complexity algebra other than $\mathcal{B}=\sum_{n=0}^{D-1} b_{n+1} [ \, \ked{\mathcal{O}_{n+1}}\brad{\mathcal{O}_n} - \ked{\mathcal{O}_{n}}\brad{\mathcal{O}_{n+1}} \, ]$, while the commutator \eqref{commBK} can only contain $\mathcal{L}=\sum_{n=0}^{D-1} b_{n+1} [ \, \ked{\mathcal{O}_{n+1}}\brad{\mathcal{O}_n} + \ked{\mathcal{O}_{n}}\brad{\mathcal{O}_{n+1}} \, ]$. Moreover, the only possibility for the algebra to be closed is that the discrete function $f(n)$ is a constant. By looking at Eq.~\eqref{f(n)}, we conclude that $f(n)$ is constant if and only if
\begin{equation} \label{closure-condition}
2(b^2_{n+1}-b^2_n)=\alpha n+2\gamma,
\end{equation}
for some constants $\alpha $ and $\gamma$ (the factors $2$ are included for convenience). Again, $b_0$ has to be replaced with $0$, so that Eq.~\eqref{closure-condition} holds for $n\geq1$, while $2b_1^2=\alpha+2\gamma$. Then, the function $f(n)$ takes the constant value $f=-\alpha /2$, so that the only possible closure of the complexity algebra is given by:
\begin{equation}\label{algebra closuree}
[\mathcal{L},\mathcal{B}]=\tilde{\mathcal{K}}, \quad [\tilde{\mathcal{K}},\mathcal{L}]=\alpha \mathcal{B}, \quad [\tilde{\mathcal{K}},\mathcal{B}]=\alpha \mathcal{L}.
\end{equation}
Moreover, from Eq.~\eqref{closure-condition} we immediately conclude that
\begin{equation}
\mathcal{\tilde{K}}=\alpha\mathcal{K}+\gamma.
\end{equation}
Therefore, if $\alpha \neq 0$,  the Krylov complexity is related to $\mathcal{\tilde{K}}$ by a shift. Conversely, if $\alpha =0$ there is no simple relation between the Krylov complexity and the operator $\tilde{\mathcal{K}}$. In this case,  $\tilde{\mathcal{K}}$ is proportional to the identity and the complexity algebra  reduces to the Heisenberg-Weyl algebra \cite{perelomov1986}, being $[\mathcal{L}_+,\mathcal{L}_-]=\gamma \mathds{1}$.

\noindent\textbf{\textsf{\small Possible scenarios under the closure of the complexity algebra.}} 
As already discussed, if $\L$, $\B$ and their commutator $\Tilde{K}$ closes an algebra, then the only possible commutation relations are given by \eqref{algebra closuree}. This complexity algebra then reduced to the Heisenberg-Weyl algebra whenever $\alpha = 0$. We next show that for the cases $\alpha <0$ and $\alpha >0$, the complexity algebra  reduces to the $\textrm{SU}(2)$ algebra and the $\textrm{SL}(2,\mathbb{R})$ algebra, respectively. Let us introduce the operators $J_+$ and $J_-$, which are defined by $\nu J_+ = \L_+$ and $\nu J_- = \L_-$, where $\nu$ is a strictly positive scaling parameter. We can then write $\L = \nu(J_+ + J_-)$ and $\B = \nu(J_+ - J_-)$. Let us also introduce the operator $J_0$ defined by $J_0 = -\frac{1}{2\nu^2}\Tilde{K}$. By substituting these operators into \eqref{algebra closuree}, one can rewrite the commutation relations as
\begin{equation}
        [J_+, J_-] = J_0,\quad [J_0, J_\pm] = \mp\frac{\alpha}{2\nu^2}J_\pm.
\end{equation}
By choosing the scaling parameter such that $2\nu^2 = \alpha$, we find that the algebra \eqref{algebra closuree} is equivalent to 
\begin{align}
     &[J_+, J_-] = J_0,\quad [J_0, J_\pm] = \pm J_\pm &\alpha< 0\quad &\textrm{SU}(2),\\
     &[J_+, J_-] = J_0,\quad [J_0, J_\pm] = \mp J_\pm &\alpha> 0\quad &\textrm{SL}(2,\mathbb{R}).
\end{align}
What we have shown is that, whenever the simplicity hypothesis holds, then the algebra generated by $\L$, $\B$ and their commutator can always be reduced to either $\textrm{SU}(2)$, $\textrm{SL}(2,\mathbb{R})$ or the Heisenberg-Weyl algebra, and for which of these it reduces to depends on the value of $\alpha$. 

 \smallskip

\noindent{\textbf{\textsf{Data availability}}}\\
The datasets generated during and/or analysed during
the current study are available from the corresponding
author upon reasonable request.

\medskip

\noindent{\textbf{\textsf{Competing interests}}}\\
 The authors declare no competing interests.

\smallskip

\noindent{\textbf{\textsf{Code availability}}}\\
The codes generated and used during the current study
are available from the corresponding author on reason-
able request.

 \smallskip
 
\noindent{\textbf{\textsf{Online content}}}\\
Methods, additional references, supplementary information, are available.
 
 \smallskip
 
\noindent{\textbf{\textsf{Acknowledgements}}}\\ We are grateful to A. Chenu and J. Yang for insightful discussions. 

\smallskip

\noindent{\textbf{\textsf{Author contributions}}}\\
 N.H. and A.D.C introduced the dispersion bound. N.H. provided its geometric interpretation and  together with N.C. studied the conditions for its saturation. N.H. provided the estimation of the deviation time. N.C. prepared Figure 1 and analyzed the explicit models presented in the supplementary information. A.S.M.-R. found the differential equation for the Krylov complexity and performed the numerical analysis in GOE for Figure 2. A.D.C proposed the project, provided guidance and supervised the work. All the authors participated in the analysis of the results and the writing of the manuscript.

\smallskip

\smallskip
 

	
 \noindent{\textbf{\textsf{References}}}\\
 \bibliography{KCbound}

 \onecolumngrid
 \pagebreak
 \section*{Ultimate Speed Limits to the Growth of Operator Complexity\\---Supplementary Information---}

\section*{Supplementary note 1: Explicit models}

In this appendix we introduce three dynamical models that, having the structure of a closed complexity algebra, display a maximal growth of complexity, in the sense that the complexity rate saturates the dispersion bound. Moreover, we show that quantum chaos, in the Hamiltonian sense, is not necessary to have maximal complexity growth. In particular, it is shown that the dynamics of simple solvable Hamiltonians can saturate our bound.

We first consider a finite-dimensional model, namely the $\textrm{SU}(2)$ algebra, and then turn to the infinite-dimensional case, which allows us to comment on the famous conjecture by Parker \textit{et al.} \cite{parker2019}. In particular, we show that our notion of maximal complexity growth is more general than the one proposed in their work, as the latter represents a special case of the former.

\subsection*{\textrm{SU}(2) algebra}

Let us start with the $\textrm{SU}(2)$ algebra $[J_i,J_j]=i\epsilon_{ijk}J_k$. That is, let us consider the dynamical evolution generated by the Liouvillian
\begin{equation}
	\mathcal{L}=\nu(J_+ + J_-),
\end{equation}
where $J_{\pm}=J_1\pm iJ_2$ are the familiar $\textrm{SU}(2)$ ladder operators and $\mathcal{L_\pm}=\nu J_\pm$. Now, the Krylov basis corresponds to the usual basis of the representation $j$: $\ked{\mathcal{O}_n}=\ket{j,n}$ with $-j \leq n \leq j$. Following \cite{caputa2021}, let us relabel the vectors with $n\to n+j$, so that $n=0,\dots,2j$, the dimension of the Krylov space being equal to $2j+1$. By construction, the initial operator $\ked{\mathcal{O}_0}$ is just the highest weight state $\ket{j,-j}$ and is annihilated by $J_-$. From the action of the ladder operators on the representation basis:
\begin{eqnarray}
	J_+ \ket{j,-j+n}&=\sqrt{(n+1)(2j-n)} \ket{j,-j+n+1},\\
	J_- \ket{j,-j+n}&=\sqrt{n(2j-n+1)} \ket{j,-j+n-1},
\end{eqnarray}
being $\mathcal{L_-}\ked{\mathcal{O}_n}=b_n \ked{\mathcal{O}_{n-1}}$, we can read off the Lanczos coefficients:
\begin{equation} \label{bn-SU2}
	b_n=\nu \sqrt{n(2j+1-n)}.
\end{equation}
The Heisenberg evolution of an operator can be understood as the displacement of a generalized coherent state \cite{caputa2021}:
\begin{equation}
	\ked{\mathcal{O}(t)}=e^{i\mathcal{L}t}\ked{\mathcal{O}_0}=D(\xi=i\nu t)\ked{\mathcal{O}_0}.
\end{equation}
Indeed, the displacement operator is defined as
\begin{equation}
	D(\xi)=e^{\xi J_+ - \xi^* J_-}.
\end{equation}
The generalized coherent state $\ket{\xi,j}$ can be expanded over the spin basis $\ket{j,-j+n}$ as follows \cite{caputa2021}:
\begin{equation}
\ket{\xi,j}=(1+|\xi|^2)^{-j} \sum_{n=0}^{2j} \xi^n \sqrt{\frac{\Gamma(2j+1)}{n! \Gamma(2j-n+1)}} \ket{j,-j+n}.
\end{equation}
For this model it is convenient to use complex polar coordinates $\xi=e^{i\phi}\tan\theta $. Indeed, by replacing $\theta=\nu t$ and $\phi=\pi/2$ and by using the correspondence between the spin and the Krylov basis $\ked{\mathcal{O}_n}=\ket{j,-j+n}$, from above one can read the components of the operator wavefunction:
\begin{equation}
\phi_n(t)= \tan^n (\nu t) \cos^{2j} (\nu t) \sqrt{\frac{\Gamma(2j+1)}{n! \Gamma(2j-n+1)}},
\end{equation}
from which we can compute both the mean (i.e.~the Krylov complexity) and the variance of the complexity operator $\mathcal{K}$:
\begin{align}
K(t)&= \sum_{n=0}^{2j} n\, \phi^2_n(t)= 2j\sin^2 \nu t,\\
\Delta \mathcal{K}(t)&= \sqrt{\sum_{n=0}^{2j} n^2 \phi^2_n(t) -K^2}=  \sqrt{\frac{j}{2}}|\sin 2\nu t|,
\end{align}
Since $b_1=\nu\sqrt{2j}$, one can check that $| \partial K| = 2b_1 \Delta \mathcal{K}$ at any time t: that is, as expected from the closure of the 3-dimensional complexity algebra, the dispersion bound is identically saturated.

Given the expression of the Liouvillian in Krylov space, it is generally a difficult task to derive a corresponding Hamiltonian that generates the dynamics in the Hilbert space. In particular, the former contains less information than the latter and therefore many different Hamiltonians can give rise to the same dynamics in Krylov space. Moreover, one has not only to specify the Hamiltonian but also the initial operator $\mathcal{O}_0$. Nevertheless, we find that the evolution of the operator $\mathcal{O}_0=\sigma_1+\sigma_3$ under the single-qubit (two-level) Hamiltonian $H=\nu \sigma_3$, where $\sigma_i$ is the $i$-th Pauli matrix, is given in Krylov space by the representation $j=1$ of the $\textrm{SU}(2)$ algebra. More precisely, by explicitly performing the Lanczos algorithm, which in this case involves only two steps, we find Lanczos coefficients $b_1=b_2=\nu\sqrt{2}$, which coincide with Eq.~\eqref{bn-SU2} for $j=1$. We note that here the dimension of the Krylov space is $D=3$, which is the maximum allowed for a Hilbert dimension $d=2$, being $D\leq d^2-d+1$ \cite{rabinovici2021}. This is achieved due to the choice made for the initial operator $\mathcal{O}_0$, which has non-zero components along all the Liouvillian eigenspaces. If instead one starts with an initial operator $\mathcal{O}_0=\sigma_1$, the Krylov dimension shrinks to $D=2$, in which case the bound is always trivially saturated, being the complexity algebra given by the representation $j=1/2$ of $\textrm{SU}(2)$.
From this example, we deduce that non-chaotic Hamiltonian can give rise to maximal complexity growth in Krylov space. Interestingly, the same observation was made also in \cite{DymarskySmolkin21} with respect to the different notion of maximal complexity growth proposed by Parker \textit{et al.} \cite{parker2019}, proving that the exponential growth of complexity can be achieved also without chaos.

As a final remark, let us note that by considering a more general two-level Hamiltonian $H=c\mathds{1}+\overrightarrow{v}\cdot\overrightarrow{\sigma}$ we can still obtain the same dynamics in Krylov space, i.e.~representation $j=1$ of $\textrm{SU}(2)$, provided that we tune the parameters and we choose the initial operator in such a way that $b_1=b_2$: if this condition does not hold, the bound cannot be saturated. More generally, in any Krylov space of dimension $D=3$, the algebraic closure and thus the saturation of the bound is possible if and only if $b_1=b_2$, i.e.~ if and only if the underlying algebra is given by the representation $j=1$ of $\textrm{SU}(2)$. This can be checked by explicitly computing the double commutator $[ \mathcal{L},[\mathcal{\tilde{K}}, \mathcal{B} ]]$ and observing that it vanishes only if $b_1=b_2$. Indeed, as emphasized in the main text, the only possible closed complexity algebra in the case of a finite Krylov dimension $D$ is given, up to a multiplicative constant, by the representation $j=\frac{D-1}{2}$ of $\textrm{SU}(2)$.

\subsection*{Heisenberg-Weyl algebra}

Let us now consider the case of infinite-dimensional Krylov space. An emblematic example in which the bound is saturated is the one in which the dynamical evolution is given in terms of the Heisenberg-Weyl ($\textrm{HW}$) algebra $[a, a^\dagger] = 1$. In this case, the Liouvillian is given by
\begin{equation}
	\mathcal{L}=\nu(a^\dagger + a),
\end{equation}
and the generalized ladder operators $\mathcal{L}_\pm$ are just the raising and lowering operators $a^\dagger$ and $a$, times the constant $\nu$. Here the initial operator $\ked{\mathcal{O}_0}$ is represented as the vacuum state $\ket{0}$ and the Krylov basis corresponds to the usual basis constructed by acting with $a^\dagger$ on the vacuum:
\begin{equation}
	\ked{\mathcal{O}_n} = \ket{n} \frac{1}{\sqrt{n!}} (a^\dagger)^n \ket{0},
\end{equation}
that is, the eigenbasis of the number operator $a^\dagger a$, which coincides with the complexity operator $\mathcal{K}$. We note that in this case the Krylov space has infinite dimension. From the well known relations
\begin{equation}
a^\dagger \ket{n} = \sqrt{n+1} \ket{n+1}, \quad a\ket{n}=\sqrt{n} \ket{n-1},
\end{equation}
one can see that $b_n=\nu \sqrt{n}$. The time-evolved operator $\ked{\mathcal{O}(t)}$ can be represented as the standard coherent state
\begin{equation}
\ket{\xi}= D(\xi) \ket{0}= e^{-|\xi|^2/2} \sum_{n=0}^{\infty} \frac{\xi^n}{\sqrt{n!}} \ket{n}
\end{equation}
for $\xi = i\nu t$. Therefore, the components of the operator wavefunction are
\begin{equation}
\phi_n(t)= e^{-(\nu t)^2/2}  \frac{(\nu t)^n}{\sqrt{n!}},
\end{equation}
from which we can compute that
\begin{equation}
K(t)=(\Delta \mathcal{K})^2= \nu^2 t^2,
\end{equation}
We thus conclude that, being $b_1 = \nu$, the dispersion bound is always saturated: that is, $| \partial K| = 2b_1 \Delta \mathcal{K}$ $\forall t$. This model provides an example in which maximal complexity growth (in the sense of saturation of our bound) is achieved, while the conjecture by Parker \textit{et al.} \cite{parker2019}, i.e.~linear growth of Lanczos coefficients, does not hold. We therefore see that the two notions of maximal complexity growth are not equivalent.

\subsection*{SYK model}

Finally, let us consider the celebrated prototype for quantum chaos: the SYK model of $N$ Majorana fermions with $q$-body interaction, given by the Hamiltonian
\begin{equation}
H^{(q)}_{SYK}= i^{q/2} \sum_{1\leq i_1 < i_2<\dots <i_q \leq N} J_{i_1\dots i_q} \gamma_{i_1}\dots \gamma_{i_q}.
\end{equation}
In the large-$N$ limit the model can be solved analytically and, for asymptotically large $q$, has been proven to obey the universal growth hypothesis by Parker \textit{et al.} \cite{parker2019}: namely, the growth of the Lanczos coefficients is asymptotically linear in $n$, resulting in the exponential time-behaviour of Krylov complexity. More precisely, it can be shown that, in this limit, the SYK belongs to a family of exact solutions with Lanczos coefficients \cite{parker2019}
\begin{equation} \label{SYK-b_n}
b_n=\nu \sqrt{n(n-1+\eta)}
\end{equation}
and amplitudes
\begin{equation}
\phi_n(t)=\sqrt{ \frac{(\eta)_n}{n!} } \tanh^n(\nu t) \sech^\eta (\nu t),
\end{equation}
where $(\eta)_n=\eta(\eta+1)\dots (\eta+n-1)$ is the Pochhammer symbol. From these amplitudes one can extract the complexity $K(t)=\eta \sinh^2(\nu t)$, which, as expected from the asymptotic linear behaviour of the Lanczos coefficients, shows an asymptotic exponential growth.
Remarkably, the linear growth of the Lanczos coefficients is a sufficient (but not necessary, as shown above) condition for the saturation of the dispersion bound on complexity, as shown in Supplementary Figure \ref{SYKplot}. This saturation is due to the presence of an underlying complexity algebra: indeed, one of the main results of our work is the proof that the closure of the complexity algebra is both a sufficient and a necessary condition for the dispersion bound to be saturated. For this particular family of solutions, the underlying algebra is that of $\textrm{SL}(2, \mathbb{R})$ \cite{caputa2021}.

  \begin{figure}[t]
  \centering
\hspace*{-0.0 cm}
\includegraphics[scale=0.52]{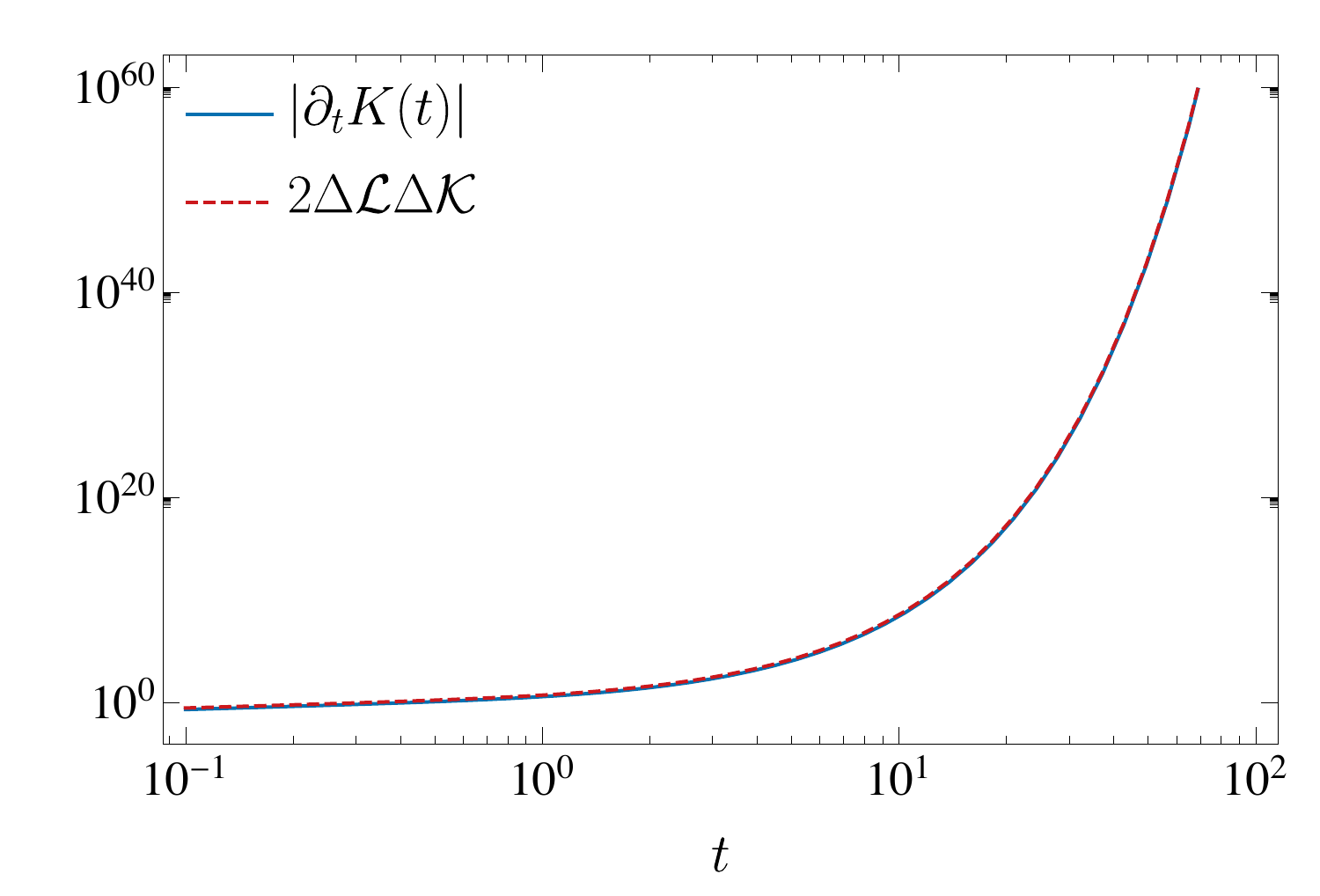}
  \caption{
  For the SYK model, the complexity rate (solid light blue line) saturates the dispersion bound (dashed red line) at any time. In the plot we fix the parameters appearing in Eq.~\eqref{SYK-b_n} to be $\nu=\eta=1$, that is we consider an exact linear growth of the Lanczos coefficients: $b_n=n$ for $n>1$.
   }
  \label{SYKplot}
    \end{figure} 

\section*{Supplementary note 2: Equivalence between the saturation of the dispersion bound and the simplicity hypothesis}

In this appendix, we show that there is an equivalence between the saturation of the dispersion bound and the simplicity hypothesis being satisfied. When we say that the complexity algebra is closed, we will simply mean that the simplicity hypothesis is satisfied.

The right-hand side of the dispersion bound is equal to two times the norm of the vectors $(\K - \expect{\K}_t)\OO(t)$ and $(\L - \expect{\L}_t)\OO(t)$, while the left-hand side is obtained by applying the Cauchy-Schwarz inequality. From this, it is clear that the bound is saturated if and only if the two vectors are linearly dependent. In other words, the bound is saturated if and only if the vectors $(\K -K)\ked{\OO(t)}$ and $\L\ked{\OO(t)}$ are linearly dependent, where we have chosen to suppress the time dependence of $K$. What will follow is a series of steps proving that the complexity algebra being closed is both necessary and sufficient for the vectors $(\K -K)\ked{\OO(t)}$ and $\L\ked{\OO(t)}$ to be linearly dependent. Said differently, the complexity algebra being closed is equivalent to the dispersion bound being saturated. When carrying out the proofs, we will use the convention that $b_0 = 0$ and for finite Krylov dimension $D$, we will also introduce $b_D=0$. For any superoperator $\M$ we will write $\M_{n,m} \equiv \braked{\OO_n}{\M\OO_m}$, where $\M_{n,m}$ can be thought of as the entries of a matrix representing $\M$.

\subsection*{Proving necessity}

Linear dependence between $(\K -K)\ked{\OO(t)}$ and $\L\ked{\OO(t)}$ is equivalent with linear dependence between $e^{-it\L}(\K -K)e^{it\L}\OO$ and $\ked{\OO_1}$. To simplify, we will use the notation $\mathbb{L}^n$ to mean $[\L,\cdot]$ applied to $\K$ $n$ times. By Taylor expanding the vector $e^{-it\L}(\K -K)e^{it\L}\ked{\OO}$ at $t=0$, we have that
\begin{equation}
    e^{-it\L}(\K-K)e^{it\L}\ked{\OO} = \Big(-K + \sum_{n=1}^{\infty}\frac{1}{n!}(-i)^n\mathbb{L}^nt^n\Big)\ked{\OO}.
\end{equation}
It is clear that $\mathbb{L} = \L_- -\L_+$ while $\mathbb{L}^2 = 2[\L_+,\L_-]$ is diagonal in the Krylov basis with eigenvalues $(\mathbb{L}^2)_{n,n} = -2(b_{n+1}^2 - b_n^2)$. Applying $[\L,\cdot]$ once more, one finds that $\mathbb{L}^3$ consists only of a subdiagonal and superdiagonal with values given by $(\mathbb{L}^3)_{n+1,n} = -(\mathbb{L}^3)_{n,n+1} = -2b_{n+1}f(n)$, where $f$ is the discrete function defined by $f(n) = (b_{n+2}^2 - b_{n+1}^2)-(b_{n+1}^2 - b_n^2)$. By the $k$-diagonal of a matrix, we mean the diagonal of the matrix going top-left to bottom-right direction where $k$ is an offset from the main diagonal. We use the convention that $k=0$ is the main diagonal while $k=1$ and $k=-1$ are the superdiagonal and subdiagonal respectively, and so on. From the form of $\mathbb{L}^3$, it should be clear that $k$-diagonals of $\mathbb{L}^{n+3}$ for which $\abs{k}>1+n$ must only consist of zero-valued entries. Consequently, we must have that $(\mathbb{L}^{n+4})_{n+m+2,m} = [\L_+,\mathbb{L}^{n+3}]_{n+m+2,m}$ which more explicitly can be written as the recursion relation $(\mathbb{L}^{n+4})_{n+m+2,m} = b_{n+m+3}(\mathbb{L}^{n+3})_{n+m+1,m}-b_{m+1}(\mathbb{L}^{n+3})_{n+m+2,m+1}$. To simplify some notation, we will write $L(n,m) \equiv (\mathbb{L}^{n+4})_{n+m+2,m}$ and the recursion relation can then be written as $L(n,m) = b_{n+m+2}L(n-1,m) - b_{m+1}L(n-1,m+1)$ for $n>0$. 

We now observe that the following proposition must be true:
\begin{prop}
\label{prop1}
The condition: $L(n,0)=0$ $\forall$ $0\leq n\leq D-3$, is a necessary condition for the vector $e^{-it\L}(\K-K)e^{it\L}\ked{\OO}$ to be linearly dependent of $\ked{\OO_1}$, and  therefore, a necessary condition for the dispersion bound to be satisfied.
\end{prop}
By applying $[\L,\cdot]$ to $\mathbb{L}^3$, one finds that $L(0,m) = 2b_{m+1}b_{m+2}g(m)$, where we have defined $g(m) = f(m)-f(m+1)$. We will show that the condition $L(n,0) = 0$ $\forall$ $0\leq n\leq D-3$ is equivalent to the complexity algebra being closed. Together with Proposition \ref{prop1}, this would then prove that the algebra being closed is a necessary condition for saturation of the dispersion bound. In order to prove this however, we will first prove another proposition.

\begin{lem}
Consider the discrete function $L(n,m)$ where $0\leq n \leq D-3$ and $0 \leq m\leq D-1$. The recursion relation $L(n,m) = b_{n+m+2}L(n-1,m) - b_{m+1}L(n-1,m+1)$ together with the initial condition $L(0,m) = 2b_{m+1}b_{m+2}g(m)$ implies:
\begin{equation}
    L(n,m) = 2\prod_{j=1+m}^{2+n+m}b_j\sum_{k=0}^n(-1)^k\binom{n}{k}g(m+k).
\end{equation}
\end{lem}

\begin{proof}
We prove this by using mathematical induction. For the base case we have that
\begin{equation}
\begin{split}
    L(1,m) &=
    b_{m+3}L(0,m) - b_{m+1}L(0,m+1)\\
    &= b_{m+3}\big(2b_{m+1}b_{m+2}g(m)\big)-b_{m+1}\big(2b_{m+2}b_{m+3}g(m+1)\big)\\
    &= 2b_{m+1}b_{m+2}b_{m+3}\big(g(m)-g(m+1)\big).
\end{split}
\end{equation}
For the inductive step we have
\begin{equation}
    \begin{split}
        L(n,m) &= b_{n+m+2}L(n-1,m) - b_{m+1}L(n-1,m+1)\\
        &= 2\prod_{j=m+1}^{n+m+2}b_j\Bigg(\sum_{k=0}^{n-1}(-1)^k\binom{n-1}{k}g(m+k)-\sum_{k=0}^{n-1}(-1)^k\binom{n-1}{k}g(m+1+k)\Bigg)\\
        &= 
        2\prod_{j=m+1}^{n+m+2}b_j\Bigg(\sum_{k=0}^{n-1}(-1)^k\binom{n-1}{k}g(m+k)+\sum_{k=1}^{n}(-1)^k\binom{n-1}{k-1}g(m+k)\Bigg)\\
        &=
        2\prod_{j=m+1}^{n+m+2}b_j\Bigg(g(m)+(-1)^{n}g(m+n)+\sum_{k=1}^{n-1}(-1)^k\Bigg[\binom{n-1}{k-1}+\binom{n-1}{k}\Bigg]g(m+k)\Bigg)\\
        &=
        2\prod_{j=m+1}^{n+m+2}b_j\sum_{k=0}^{n}(-1)^k\binom{n}{k}g(m+k),
    \end{split}
\end{equation}
where, in obtaining the second last line, we have made use of the binomial identity $\binom{n}{k} = \binom{n-1}{k-1} + \binom{n-1}{k}$.
\end{proof}

\begin{cor}
\label{cor1}
    $ L(n,0) = 2\prod_{j=1}^{2+n}b_j\sum_{k=0}^n(-1)^k\binom{n}{k}g(k)$.
\end{cor}

\begin{prop}
\label{prop2}
    $L(n,0)=0$ $\forall$ $0\leq n\leq D-3$ $\Leftrightarrow$ $g(n)=0$ $\forall$ $0\leq n\leq D-3$.
\end{prop}
\begin{proof}
    We can think of the set of functions $g(n)$ as spanning a subset of $\mathbb{R}^{D-3}$. It should then be clear from Corollary \ref{cor1} that the set of functions $L(n,0)$ must then have the same span. This means that we can express each $g(n)$ as a linear combination of the functions $L(n,0)$ or vice versa. Equating each function $L(n,0)$ ($g(n)$) with zero then results in $g(n) = 0$ ($L(n,0)=0$) for all $0\leq n\leq D-3$.
\end{proof}
We are now ready to prove the following proposition:
\begin{prop}
\label{prop nec}
The saturation of the dispersion bound implies that the complexity algebra is closed.
\end{prop}

\begin{proof}
We have that $\B \equiv \mathbb{L}$ and $\Tilde{\K} \equiv \mathbb{L}^2$ and the complexity algebra is closed per definition if and only if $\mathbb{L}^3 = [\L,\Tilde{K}]$ can only be written as a linear combination of $\L$, $\B$ and $\Tilde{\K}$. It should be clear that this is possible if and only if $f(n)=C$ $\forall$ $0\leq n\leq D-2$, where $C\in\mathbb{R}$. This is clearly equivalent to the condition $g(n) = 0$ $\forall$ $0\leq n\leq D-3$, which together with Proposition \ref{prop1} and \ref{prop2} is implied by saturation of the dispersion bound.
\end{proof}

\subsection*{Proving sufficiency}

As we pointed out in the proof of Proposition \ref{prop nec}, the complexity algebra being closed is equivalent with $f(n)=C$ $\forall$ $0\leq n\leq D-2$, where $C\in\mathbb{R}$. We note that
\begin{align}
\label{appendix f(n)}
    f(n)=C\textrm{ }\forall\textrm{ }0\leq n\leq D-2\textrm{ }&\Leftrightarrow\textrm{ }  2(b_{n+1}^2-b_n^2) = \alpha n + \gamma \quad\forall\textrm{ }0\leq n\leq D-1\\
    &\label{appendix b_n}\Leftrightarrow\textrm{ }b_{n} = \sqrt{\frac{1}{4}\alpha n(n-1)+\frac{1}{2}\gamma n + \delta} \quad\forall\textrm{ }0\leq n\leq D,
\end{align}
where $\alpha$, $\gamma$ and $\delta$ are real constants and $C=\frac{1}{2}\alpha$. We stress that \eqref{appendix b_n} holds under the convention that $b_0 = 0$, and we note that this implies that $\delta = 0$ and so we must have
\begin{equation}
    b_{n} = \sqrt{\frac{1}{4}\alpha n(n-1)+\frac{1}{2}\gamma n} \quad\forall\textrm{ }1\leq n\leq D.
\end{equation}
The right hand side of the equivalence sign in \eqref{appendix f(n)} is equivalent to $\Tilde{K} = \alpha\K +\gamma$. Consequently, the closed complexity algebra is entirely determined by the commutation relations $[\K,\L] = \B$, $[\K,\B] = \L$ and $[\L,\B] = \alpha\K + \gamma$.

\begin{lem}
\label{lemma 2}
The complexity algebra being closed implies that $\mathbb{L}^{2n} = (-1)^n\frac{1}{\alpha}(\sqrt{\alpha})^{2n}(\alpha\K + \gamma)$ and $\mathbb{L}^{2n+1} = (-1)^{n+1}\frac{1}{\sqrt{\alpha}}(\sqrt{\alpha})^{2n+1}\B$ when $\alpha\neq 0$ and $\mathbb{L} = -\B$, $\mathbb{L}^2 = -\gamma$ and $\mathbb{L}^n=0$ for $n>2$ when $\alpha = 0$.
\end{lem}
\begin{proof}
The case for when $\alpha=0$ is trivial while for $\alpha\neq 0$ we will use mathematical induction.
For the base case we have $\mathbb{L}^2 = [\L,[\L,\K]] = -[\L,\B] = -(\alpha\K + \gamma)$ and $\mathbb{L}^3 = [\L, -(\alpha\K + \gamma)] = \alpha\B$. For the inductive step, we have $\mathbb{L}^{2n} = [\L,[\L,\mathbb{L}^{2(n-1)}]] = (-1)^{n-1}(\sqrt{\alpha})^{2(n-1)}[\L,[\L,\K]] = (-1)^n\frac{1}{\alpha}(\sqrt{\alpha})^{2n}(\alpha\K + \gamma)$ and $\mathbb{L}^{2n+1} = [\L,\mathbb{L}^{2n}] = (-1)^{n+1}\frac{1}{\sqrt{\alpha}}(\sqrt{\alpha})^{2n+1}\B$.

\end{proof}

\begin{prop}
\label{prop suff}
The complexity algebra being closed implies that the dispersion bound is saturated.
\end{prop}
\begin{proof}
By Lemma \ref{lemma 2} we have $\mathbb{L}^{2n} = (-1)^n\frac{1}{\alpha}(\sqrt{\alpha})^{2n}(\alpha\K + \gamma)$ and $\mathbb{L}^{2n+1} = (-1)^{n+1}\frac{1}{\sqrt{\alpha}}(\sqrt{\alpha})^{2n+1}\B$ when $\alpha\neq 0$. By substituting these into the Taylor expansion of $e^{-it\L}(\K-K)e^{it\L}\ked{\OO}$, we have
\begin{equation}
\label{expansion1}
    \begin{split}
        e^{-it\L}(\K-K)e^{it\L}\ked{\OO} &= \Big(-K + \sum_{n=1}^{\infty}\frac{1}{n!}(-i)^n\mathbb{L}^nt^n\Big)\ked{\OO}\\
        &= \Big(-K + \sum_{n=1}^{\infty}\frac{1}{(2n)!}(-i\mathbb{L}t)^{2n} + \sum_{n=0}^{\infty}\frac{1}{(2n+1)!}(-i\mathbb{L}t)^{2n+1}\Big)\ked{\OO}\\
        &= \Big(-K + \frac{\gamma}{\alpha}\sum_{n=1}^{\infty}\frac{1}{(2n)!}(\sqrt{\alpha}t)^{2n} + \frac{1}{\sqrt{\alpha}}i\B\sum_{n=0}^{\infty}\frac{1}{(2n+1)!}(\sqrt{\alpha}t)^{2n+1}\Big)\ked{\OO}\\
        &= \Big(-K + \frac{\gamma}{\alpha}(\cosh{\sqrt{\alpha }t}-1) + \frac{1}{\sqrt{\alpha}}i\B\sinh{\sqrt{\alpha}t}\Big)\ked{\OO}.
    \end{split}
\end{equation}
Since $\B\ked{\OO} = b_1\ked{\OO_1}$, it follows from the definition of Krylov complexity that the first two terms in the expression above must cancel. We thus have that
\begin{equation}
    e^{-it\L}(\K-K)e^{it\L}\ked{\OO} = \frac{b_1}{\sqrt{\alpha}}\sinh{\sqrt{\alpha}t}\ked{\OO_1}.
\end{equation}
When $\alpha = 0$, one has that $\mathbb{L} = -\B$, $\mathbb{L}^2 = -\gamma$ and $\mathbb{L}^n=0$ for $n>2$. Substituting these into the Taylor expansion, one finds that
\begin{equation}
\label{expansion2}
    \begin{split}
        e^{-it\L}(\K-K)e^{it\L}\ked{\OO} &= \Big(-K-i\mathbb{L}t - \frac{1}{2}\mathbb{L}^2t^2\Big)\ked{\OO}\\
        &= \Big(-K+\frac{\gamma}{2}t^2 +i\B t\Big)\ked{\OO} = ib_1\ked{\OO_1}.
    \end{split}
\end{equation}
We thus have that the algebra being closed is a sufficient requirement for saturating the dispersion bound.
\end{proof}
The proofs of Proposition \ref{prop nec} and \ref{prop suff} leads to the conclusion that saturation of the dispersion bound is equivalent with the complexity algebra being closed.
\begin{remark}
    We would like to point out that equation \eqref{expansion1} and \eqref{expansion2} shows that the general solution for Krylov complexity, whenever the dispersion bound is saturated, is given by  $K(t) = -\frac{2\gamma}{\alpha}\sin^2{\frac{\sqrt{-\alpha}t}{2}}$ when $\alpha< 0$, $K(t) = \frac{\gamma}{2}t^2$ when $\alpha = 0$ and $K(t) = \frac{2\gamma}{\alpha}\sinh^2{\frac{\sqrt{\alpha}t}{2}}$ when $\alpha> 0$. These three scenarios correspond to the three algebraic models discussed above:  $\textrm{SL}(2, \mathbb{R})$, $\textrm{HW}$ and $\textrm{SU}(2)$ respectively.
\end{remark}
\begin{remark}
    The requirement that $b_n\geq 0$ for all $n$ implies that $\gamma\geq 0$ and $\alpha\geq -\frac{2}{n-1}\gamma$ for all $n$. In the infinite dimensional case we see that this implies that $\alpha\geq 0$. In the finite-dimensional case, the condition $b_D=0$ implies that $\alpha = -\frac{2}{D-1}\gamma$ and so the solution of Krylov complexity only depends on $\gamma$, namely $K(t) = (D-1)\sin^2{\sqrt{\frac{\gamma}{2(D-1)}}t}$. By setting $\omega = \sqrt{\frac{\gamma}{2(D-1)}}$ we have that $K(t) = (D-1)\sin^2{\omega t}$ and the Lanczos coefficients grow according to $b_n = \omega\sqrt{n(D-n)}$. Therefore, by comparison with Eq.~\eqref{bn-SU2}, we see that, in a finite $D$-dimensional Krylov space, the saturation of the bound for each time can be achieved only when the dynamics is governed, up to a multiplicative constant, by the $\textrm{SU}(2)$ algebra, in the representation $j=\frac{D-1}{2}$.
\end{remark}

\end{document}